\documentclass{amsart}
\usepackage{bbm}
\usepackage[latin1]{inputenc}
\usepackage[xdvi]{graphicx}

\evensidemargin \oddsidemargin
\newcommand\ind[1]{\mathbbm{1}_{\{#1\}}}
\newcommand\ovl[1]{\overline{#1}}
\newcommand\N{\mathbb{N}}
\newcommand\R{\mathbb{R}}
\newcommand\DD{\mathcal{D}}
\newcommand\E{\mathbb{E}}
\renewcommand\P{\mathbb{P}}
\newcommand\eps{\varepsilon}

\newtheorem{theorem}{Theorem}[section]
\newtheorem{proposition}[theorem]{Proposition}
\newtheorem{corollary}[theorem]{Corollary}

\newtheorem{definition}{Definition}
\newtheorem{remark}{Remark}

\def\cal{\mathcal}

\title[Multi-class intermittent connections]{Self-adaptive congestion control for 
multi-class intermittent connections in a communication network}

\author{Carl Graham}
\address[C. Graham]{UMR 7641  CNRS --- École Polytechnique, Route de Saclay, 91128 Palaiseau, France}
\email{carl@cmapx.polytechnique.fr}

\author{Philippe Robert}
\address[Ph.~Robert]{INRIA Paris --- Rocquencourt, Domaine de Voluceau, 78153 Le Chesnay, France}
\email{Philippe.Robert@inria.fr}
\urladdr{http://www-rocq.inria.fr/\string~robert}

\date{August 20, 2010}
\keywords{Flow control algorithm \and Multi-class system \and Mean-field limit \and Chaoticity \and 
Nonlinear stochastic differential equation \and Stationary distribution}

\begin{document}

\maketitle

\begin{abstract}
A Markovian  model of the  evolution of intermittent  connections of various classes  in a
communication network  is established and investigated.   Any connection evolves  in a way
which depends  only on its  class and the  state of the network,  in particular as  to the
route it uses among a subset of the network nodes. It can be either active (ON) when it is
transmitting data  along its  route, or  idle (OFF).  The  congestion of  a given  node is
defined as a functional of the transmission  rates of all ON connections going through it,
and causes  losses and  delays to  these connections.  In  order to  control this,  the ON
connections  self-adaptively  vary  their  transmission  rate in  TCP-like  fashion.   The
connections interact  through this feedback  loop.  A Markovian  model is provided  by the
states  (OFF, or  ON  with some  transmission rate)  of  the connections.   The number  of
connections in each class being potentially huge, a mean-field limit result is proved with
an appropriate scaling so as to reduce the dimensionality.  In the limit, the evolution of
the states  of the  connections can be  represented by  a non-linear system  of stochastic
differential equations,  of dimension  the number of  classes.  Additionally, it  is shown
that  the corresponding  stationary distribution  can be  expressed by  the solution  of a
fixed-point equation of finite dimension.
\end{abstract}

\section{Introduction}
The  Internet network  can be  described as  a very  large distributed
system,  which  manages  the  dynamic exchange  of  data  transmission
through connections  established between nodes of the  network under a
very dynamic, random, scheme. Nodes  cannot cope at all times with the
huge  amount   of  data  transmitted  through  them   by  the  varying
connections, and congestion events occur, causing losses and delays to
the connections.

AIMD  algorithms (Additive Increase,  Multiplicative Decrease)  in the
TCP  protocol  regulate   Internet  traffic:  a  connection  increases
gradually its  throughput as long  as congestion does not  occur along
its route (Additive Increase), but  decreases it brutally when such an
event  takes  place,  usually   by  cutting  it in  half
(Multiplicative Decrease).

An important  practical issue  is to obtain  a better  qualitative and
quantitative  understanding  of  such  algorithms,  and  devising  and
analyzing pertinent  mathematical models  may be very  helpful.  There
are  many serious challenges,  see Graham  and Robert~\cite{Graham:09}
for a discussion of the literature in this domain. Some objectives are
to:
\begin{enumerate}
\item Propose a stochastic model for the arrivals and durations of connections.
\item Propose a stochastic model of flow control by AIMD algorithms; see, \emph{e.g.}, Dumas \emph{et al.}~\cite{Dumas:06} and Guillemin \emph{et al.}~\cite{Guillemin:05}. 
\item  Describe the state of the network, including the large number of connections interacting at its nodes through the congestion they create, in a mathematically tractable way. 
\end{enumerate}

Let us be more specific.  A client at a node initiates a connection to
a server at another node, in  order to transmit data for some specific
purpose.   Each  data packet  is  routed  from  source to  destination
through intermediate nodes, using local routing decisions depending on
the instantaneous state of the  network.  The connection adapts to the
congestion encountered  by its packets by adjusting  its throughput as
above, and ends once all the data has been completely transferred.

The same client may well repeat these steps at various instances with the same server, purpose, and reaction characteristics to network conditions, and thus intermittently establish connections of the same kind.
For simplicity of expression, we call ``user'' such a client-server-purpose-reaction combination.

Graham and Robert~\cite{Graham:09} introduced and analyzed a Markovian network model, 
further studied in Graham \emph{et al.}~\cite{GrahamRV09},
constituted
of  $J$ nodes and hosting $K$ classes of permanent connections.
This paper extends this to $K$ classes of users establishing intermittent connections
in an ON-OFF  pattern. 
An OFF user is idle, and switches ON after a random duration.  
While ON, the user transmits under an AIMD scheme (after establishing a connection)
and switches OFF (cutting the connection) after a random duration depending on 
throughput (related to node load).  
In particular, the number of on-going connections varies with time.

This gives a possible answer to Point~1 above; ON-OFF users are quite natural in models for data transmission in the Internet, for instance a  web transaction (http  connection) can be thought of as a succession of file transfers between server and client. 
For Point~2, the load
of a node may be defined as a linear functional of the throughputs of all connections going
through it, and the loss rate as a function of these quantities, as in~\cite{Graham:09}.  

Point~3 is achieved by proving a mean-field limit result, propagation of chaos,
which we describe roughly. 
When the numbers of users in each class converge to infinity with limit ratios,
for adequate initial conditions, asymptotically users evolve independently,
and the limit behavior  of users of class  $k$ for $1\leq k\leq  K$ is
described by a nonlinear (McKean-Vlasov) stochastic differential  equation (SDE)  in
$\R_+^K$, the effect of the interactions on the connections 
being encoded in the non-linearity of the SDE, as in~\cite{Graham:09}.

The  introduction of intermittent connections  significantly changes  the  problems
investigated in~\cite{Graham:09,GrahamRV09}.  A cemetery state $-1$ is added to encode
the  fact that a user can be OFF, and notably creates technical problems  concerning the
mandatory Lipschitz properties for mean-field convergence.  Moreover, the
limit stationary distributions are still
characterized by finite vectors solving a fixed point  equation, but
their expression is  much more complicated than before; it  is related to the
resolvent of a generic one-dimensional Markov process.

Section~\ref{secMod} introduces more precisely the stochastic model of the network, 
and the mean-field scaling.
Section~\ref{secLimit} gives the main properties of the limiting nonlinear SDE.
Section~\ref{secMF} establishes the  mean-field result.
Section~\ref{secEqui} studies equilibrium properties of the limit.

\noindent
{\bf Acknowledgements}\\
This work has been presented at the Stochastic Networks workshop in Cambridge in March
2010 in the programme on communication networks at the Newton Institute. We are grateful
to the organizers for the invitation. 

\section{Markovian multi-class model for the network}\label{secMod}

\subsection{Classes of ON-OFF users, throughput, node load}

Recall that ``user'' denotes a possible kind of network use by a source-destination pair,
and hence corresponds to the behavior, with respect to congestion, of its packet routes, 
rate of increase of throughput, drop
rate, etc.

The stochastic network 
is constituted of $J\ge1$ nodes and hosts $K\ge1$ classes of users.
There are $N_k\geq 1$ users of class $k$ for $1\leq k\leq K$, which bounds the \emph{varying} number of possible on-going transmissions of this class.
The notations
\[
N=(N_1,\ldots,N_K)\,, \qquad |N| = N_1 + \cdots + N_K\,,
\] 
will be used,
and $|N|$ is the total number of users in the network.

Users alternate between being \emph{active}, or ON, and
\emph{inactive}, or OFF. 
An ON user has data to transmit, and establishes a connection in the network, with throughput 
controlled by a dynamically self-managed ``window size" with values in $\R_+$ until it is done. A fictitious window size of $-1$ will be assigned to OFF users.
The state space for a user is $\R_+\cup\{-1\}$, with the trace of the usual metric.
The cemetery state (with resurrections) 
brings new problems
with respect to Graham and Robert~\cite{Graham:09}, but the choice of $-1$ simplifies notations.

The way ON users utilise the nodes depends only 
on their classes,
and is given in terms of an \emph{allocation  matrix} 
\[
A = (A_{jk}, 1  \le j \le J,\  1 \le k \le  K)\,,
\qquad 
A_{jk}\in\R_+\,,
\]
as follows:
at node $j$, the instantaneous weighted throughput of a class $k$ user
in state $w \in \R_+ \cup\{-1\}$
is given by
\[
A_{jk} w^+ := A_{jk} w\ind{w \ge 0}\,,
\]
which vanishes for OFF users and is proportional to state (window size) for ON users.

Hence the load at node $j$ is given by the total weighted throughput of all users,
and when the $n$-th of class $k$ is in state $w_{n,k}$, by 
\begin{equation}\label{nwt}
u_j = \sum_{k=1}^K\sum_{n=1}^{N_k}  A_{jk} w_{n,k}^+\,,
\qquad
1\le j \le J\,.
\end {equation}
The node load vector
\[
u =  (u_j, {1 \le  j \le J})
\] 
is an important descriptor of the congestion of the network, and is indirectly sensed by the users through their losses.
It is a linear function of the states of the ON users. 

In a simple example, mapping classes to routes,
\[
A_{jk} = 1 \;\;\text{if node $j$ is used by a class $k$ user, and}\;\; A_{jk} = 0
\;\;\text{otherwise.}
\]
More general matrices allow to model, \emph{e.g.},
different utilisations of a given node by 
different classes, or
packets of a given user
taking varied routes in the network.

\subsection{Connection initiation and termination,  congestion control}

Transitions between ON and OFF states are given by
functions $\lambda_k:\R_+^J \rightarrow \R_+$
and $\mu_k : \R_+ \times \R_+^J \rightarrow \R_+$ and laws $\alpha_k$ on $\R_+$, for $1\le k \le K$.
When the node load vector is in state $u =  (u_j, {1 \le  j \le J})$, 
see \eqref{nwt}, then in class $k$:
\begin{itemize}
\item
any OFF user (in state $-1$) is turned ON at instantaneous rate $\lambda_k(u)$ 
with new state $w$ chosen according to $\alpha_k(dw)$, and establishes a connection for data transmission,
\item
any ON user in state $w\ge0$ is turned OFF (to state $-1$)
at rate $\mu_k(w,u)$, and terminates the connection.
\end{itemize}

The throughput of connections are determined by the instantaneous states of the users that 
have initiated them; in a simple example, their initial laws $\alpha_k$ are delta masses, such as $\delta_0$.
They cause congestion at the nodes of the network, see \eqref{nwt}, which is   
controlled by ON users by regulating their state (throughput). This involves
functions $a_k: \R_+ \times \R_+^J \rightarrow \R_+$ and $b_k: \R_+ \times \R_+^J \rightarrow \R_+$
and constants $r_k\in [0,1]$, for $1\le k\le K$;
when the node load vector is $u =  (u_j, {1 \le  j \le J})$ 
and the state of a class $k$ user is $w\ge0$, then instantaneously:
\begin{itemize}
\item  the state increases continuously at speed $a_k(w, u)$,
\item the state jumps from $w$ to $r_k w$ at rate $b_k(w, u)$, corresponding to the occurrence of a
loss in the transmission due to congestion (usually $r_k=1/2$). 
\end{itemize}

The study in Graham and Robert~\cite{Graham:09} considered permanent
connections, corresponding to having $\mu_k(w,u)\equiv0$ and all users ON
at time $0$. Then, users can be identified with connections and have state space $\R_+$, 
and there is no use for $\lambda_k$ and $\alpha_k$.

\subsubsection*{Natural special forms for the parameters}

A simple case is when the throughput rate is proportional to the state (window size) of any user,
so that there are functions 
$\nu_k$ and $\beta_k$ from $\R_+^J$ to $\R_+$ such that 
\begin{equation} 
\label{outputprop}
\mu_k(w,u) = w \nu_k(u)\,,
\qquad
b_k(w,u) = w \beta_k(u)\,.
\end{equation}
A special case has constant functions
$\lambda_k(u) \equiv \lambda_k>0$
and
$\nu_k(u) \equiv \nu_k>0$
(with slight abuse of notation),
so that any class $k$ user goes ON at rate $\lambda_k$, is required to transmit a quantity of data which follows an exponential law with parameter $\nu_k$ at instantaneous throughput rate given by its window, and then goes OFF. 
Also, a natural case has 
\begin{equation}
\label{incrnow}
a_k(w,u) = a_k(u)
\end{equation}
regardless of $w$ (with slight abuse of notation), and
related to the inverse of the round-trip time (RTT) in the network when its load
is given by the  vector $u$. 

A natural sub-case of the above  is when
\begin{equation}
\label{bafopa}
a_k(u)
= \biggl( \tau_k + \sum_{j=1}^J t_{jk} (u_j) \biggr)^{-1}\,,
\qquad
\beta_k(u) = \delta_k +\sum_{j=1}^J d_{jk} (u_j)\,,
\end{equation}
where, for class $k$ users, $\tau_k>0$ is the RTT between source and destination and $\delta_k\ge0$ is the loss rate in a non-congested  network, and $t_{jk}(u_j)\ge0$ is the 
supplementary RTT delay and $d_{jk}(u_j)\ge0$ loss rate at node $j$ when its load is $u_j$. We may expect $t_{jk}(u_j)$ and $d_{jk}(u_j)$ to behave linearly in $u_j$, at least for large $u_j$.

Hence, care will be taken to include cases where $b_k(w,u)$ may have a \emph{quadratic} behavior in the assumptions to be made.

General forms of the parameters $\lambda_k$, $\mu_k$, $\alpha_k$, $a_k$, $b_k$, and $r_k$  
are used to model, \emph{e.g.}, balking (at arrival and later) in a congested network, more complex relations between window size and throughput  and losses, sequences of losses
due to congestion, slow start, RED, etc. 
For instance, a last very simple case is  $\lambda_k(u) \equiv \lambda_k>0$
and $\mu_k(w,u) \equiv \mu_k>0$, corresponding to someone popping in and out
to browse the Internet in his spare time.

\subsection{Markov process and its SDE representation,  mean-field scaling}

A Markov process describing this model is given  by
\[
W^N(t)=(W_{n,k}^N(t), {1\le n \le N_k,\, 1\le k \le K})\,,
\qquad
t\ge0\,,
\] 
where 
\[
W_{n,k}^N(t)\in \R_+\cup\{-1\}
\]
is the state of the $n$-th user of class~$k$ at time $t$.
We choose to represent it as the solution
of the Itô-Skorohod stochastic differential equation (SDE)
\begin{equation}
\label{2sde}
\left\{
\begin{aligned}
& dW_{n,k}^N(t) =
\ind{W_{n,k}^N(t-) = -1}
\int (1+w) \ind{0 < z < \lambda_k({U}^N(t-))} \,\mathcal{A}_{n,k}(dw,dz,dt)
\\ 
&\hphantom{dW_{n,k}^N(t) =}
\;\;
+ \ind{W_{n,k}^N(t-)\ge 0}
\biggl[
a_k(W_{n,k}^N(t-), {U}^N(t-))\,dt
\\ 
&\hphantom{dW_{n,k}^N(t) =}
\qquad
- (1-r_k) W_{n,k}^N(t-) 
\int \ind{0 < z <  b_k(W_{n,k}^N(t-), {U}^N(t-))} \,\mathcal{N}_{n,k}(dz,dt)
\\ 
&\hphantom{dW_{n,k}^N(t) =}
\qquad
- (1 + W_{n,k}^N(t-))
\int \ind{0 < z <  \mu_k (W_{n,k}^N(t-), {U}^N(t-))} \,\mathcal{D}_{n,k}(dz,dt)
\biggr]\,,
\\
&1\leq k\leq K\,,\; 1\leq n\leq N_k\,;
\\
& U^N(t) = (U_j^N(t), {1 \le j \le J})\,,
\qquad
{U}^N_j(t) = \sum_{k=1}^K A_{jk} \sum_{i=1}^{N_k}{ W}^N_{n,k}(t)^+\,,
\end{aligned}
\right.
\end{equation}
driven by Poisson point processes $\mathcal{A}_{n,k}$ 
on $\R_+^3$ with intensity measure $\alpha_k(dw)\,dz\,dt$ and 
$\mathcal{N}_{n,k}$ and $\mathcal{D}_{n,k}$ on $\R^2_+$ with intensity measure $dz\,dt$, 
all independent. See Graham and Robert~\cite{Graham:09} for a discussion of the related
simpler model
for permanent connections.
The following can be proved with standard arguments. 

\begin{proposition}\label{eusde}
If the functions $a_k$ are Lipschitz
and the functions $b_k$, $\lambda_k$, and $\mu_k$ are locally bounded, $1\leq k\leq K$, then there is pathwise existence and uniqueness of solution for the stochastic
differential equation~\eqref{2sde}, and the corresponding Markov process is well defined.
\end{proposition}

\subsection{The mean-field asymptotic regime}

The high-dimensional system of coupled stochastic differential equations~\eqref{2sde} does  not seem to be 
mathematically tractable as such, and an asymptotic study,
reducing its dimension $|N|$ to the number of classes $K$,
is used to investigate the qualitative
properties of solutions.

We consider the mean-field scaling in which
\begin{equation}\label{e:mfar}
N_k\to\infty\,,
\quad
\frac{N_k}{|N|} := \frac{N_k}{N_1 + \cdots + N_K} \to p_k\,,
\qquad
1\leq k\leq K\,,
\end{equation}
(necessarily $(p_k)_{1\le k \le K}$ is a probability vector),
and moreover the capacities of the resources are scaled by a factor $|N|$, 
\emph{i.e.}, $U^N$ is replaced by $\ovl{U}^{N} = \frac1{|N|}U^N$ in \eqref{2sde}
so as to have a non-trivial limit.
This procedure leads to the \emph{rescaled} SDE 
\begin{equation}
\label{sde}
\left\{
\begin{aligned}
& dW_{n,k}^N(t) =
\ind{W_{n,k}^N(t-) = -1}
\int (1+w) \ind{0 < z < \lambda_k(\ovl{U}^N(t-))} \,\mathcal{A}_{n,k}(dw,dz,dt)
\\ 
&\hphantom{dW_{n,k}^N(t) =}
\;\;
+ \ind{W_{n,k}^N(t-)\ge 0}
\biggl[
a_k\bigl(W_{n,k}^N(t-), \ovl{U}^N(t-)\bigr)\,dt
\\ 
&\hphantom{dW_{n,k}^N(t) =}
\qquad
- (1-r_k) W_{n,k}^N(t-) 
\int \ind{0 < z <  b_k(W_{n,k}^N(t-), \ovl{U}^N(t-))} \,\mathcal{N}_{n,k}(dz,dt)
\\ 
&\hphantom{dW_{n,k}^N(t) =}
\qquad
- (1 + W_{n,k}^N(t-))
\int \ind{0 < z <  \mu_k (W_{n,k}^N(t-), \ovl{U}^N(t-))} \,\mathcal{D}_{n,k}(dz,dt)
\biggr]\,,
\\
&1\leq k\leq K\,,\; 1\leq n\leq N_k\,;
\\
& {\ovl{U}}^N(t) 
= \bigl(\ovl{U}_j^N(t), {1 \le j \le J}\bigr) \,,
\qquad
\ovl{U}^{N}_j(t) = \frac{1}{|N|} U^N_j(t)
= \sum_{k=1}^K  A_{jk} \frac{N_k}{|N|} \ovl{W}^{N}_k(t)\,,
\\
&\ovl{W}^N_k(t) = \frac{1}{N_k}\sum_{n=1}^{N_k} W_{n,k}^N(t)^+\,.
\end{aligned}
\right.
\end{equation}

This is a multi-class mean-field system, in interaction through the 
vector 
\[
\ovl{W}^N(t) = \bigl(\ovl{W}^N_k(t), {1\le k \le K}\bigr)
\]
of the empirical means of the class throughputs, \emph{via}
the vector $\ovl{U}^N\!(t)$ of the \emph{scaled loads}.
For  $1\leq k\leq K$,  it is  natural to  introduce the  \emph{empirical measure}  for the
processes associated to class~$k$ users, given by
\[
\Lambda^N_k = \frac{1 }{N_k} \sum_{n=1}^{N_k} \delta_{(W_{n,k}^N(t), t\geq 0)}
\] 
where $\delta_{(x(t),t\geq 0)}$ denotes the Dirac mass at the function $t\in \R_+ \mapsto x(t)\in\R_+\cup\{-1\}$, and
in particular $\ovl{W}^N_k(t) = \langle w^+, \Lambda^N_k(t)(dw)\rangle$,
where $\Lambda^N_k(t)$ is the marginal of $\Lambda^N_k$ at time $t$.

\subsection{Notations and conventions}
For $x=(x_m, {1\le m \le M})\in \R^M$ and $v:\R_+\to \R^M$ for some $M\in \N$ and $T>0$,
let
\[
\|x\| =\max_{1\leq m\leq M}|x_m|\,,
\qquad
\|v\|_T = \sup_{0\leq s\leq T} \|v(s)\|\,.
\]
Depending on the context, the function $v = (v(t), {t\ge 0})$ may be denoted
\[
(v_m(t), {1\leq m \leq M}, {t\geq 0}) 
\;\;\text{or}\;\; 
(v_m(t), {1\leq m \leq M})
\;\;\text{or}\;\;
(v(t))
\,.
\]

If $H$ is a complete metric space, $\DD(\R_+,H)$ denotes the Skorohod space of functions
with values in $H$, right-continuous with left limits at any point of $\R_+$, see
Billingsley~\cite{Billingsley:02}.

The r.v.\ $\Lambda^N_k$ has values in the
set $\mathcal{P}(\DD(\R_+,\R_+\cup\{-1\}))$ of probability measures on $\DD(\R_+,\R_+\cup\{-1\})$.
  
\section{Analysis of the nonlinear limit process}\label{secLimit}

\subsection{Heuristic derivation of the mean-field limit}

For~\eqref{sde}, the symmetry properties within each user class lead us to expect a mean-field convergence phenomenon as $N$ gets large,
for appropriately converging initial conditions:
\begin{itemize}
\item
the processes $W_{n,k}^N$ for $1\le k\le K$ and $1\le n \le N_k$ \emph{should} become independent and converge in law to processes $W_k$ (depending only on the class $k$), 
\item
and the empirical measure $\Lambda^N_k$ {\em should}
converge in law to a deterministic limit given by the law of the same process $W_k$,
\end{itemize}
where the stochastic process 
\[
(W(t), {t\geq 0})=((W_k(t), {t\geq 0}), {1\leq k\leq K})
\]
is the solution of the nonlinear, or McKean-Vlasov, Itô-Skorohod SDE
\begin{equation}
\label{nlsde}
\left\{
\begin{aligned}
& dW_k(t) =
\ind{W_k(t-) = -1}
\int (1+w) \ind{0 < z < \lambda_k\left(u_W(t)\right)}
\,\mathcal{A}_k(dw,dz,dt)
\\ 
& \hphantom{dW_k(t) =}
\;
+ \ind{W_k(t-)\ge 0} \biggl[ a_k\left(W_k(t-), u_W(t)\right)\,dt 
\\ 
& \hphantom{dW_k(t) =}
\qquad 
- (1-r_k) W_k(t-) 
\int \ind{0 < z < b_k\left(W_k(t-), u_W(t)\right)}
\,\mathcal{N}_k(dz,dt)
\\
& \hphantom{dW_k(t) =}
\qquad 
- (1 + W_k(t-))
\int \ind{0 < z < \mu_k\left(W_k(t-), u_W(t)\right)}
\,\mathcal{D}_k(dz,dt)
\biggr]\,,
\\
&1\le k \le K\,;
\\
& u_W(t) = (u_{W,j}(t), 1 \le j \le J)\,,
\qquad
u_{W,j}(t) = \sum_{k=1}^K A_{jk} p_k \E(W_k(t)^+)\,,
\end{aligned}
\right.
\end{equation}
driven by Poisson point processes $\mathcal{A}_k$ on $\R_+^3$ with 
intensity measure $\alpha_k(dw)\,dz\,dt$ 
and $\mathcal{N}_k$ and $\mathcal{D}_k$ on $\R_+^2$ with intensity measure $dz\,dt$, all independent.

The interaction  between coordinates depends on 
the  mean throughput vector 
\[(u_W(t), t\ge0)\]
which is a linear
functional of 
the mean class output vector
\[
\E(W(t)^+) = \E(W_k(t)^+, {1\leq k\leq K}) = \langle w^+, \mathcal{L}(W(t)) \rangle\,.
\]
In  particular, the  infinitesimal  generator of  the process  $(W(t),
{t\ge0})$ depends,  at time $t$, on  the \emph{law} of  $W(t)$ and not
only  on the value  taken by  the sample  path, as  it is  usually the
case. Using  this generator, there  is a nonlinear  martingale problem
formulation for the weak interpretation of Equation~\eqref{nlsde}.

Some properties of the solutions of the SDE~\eqref{nlsde} are analyzed
in this section. The mean-field  convergence results will be the topic
of Section~\ref{secMF}.

\subsection{Existence and uniqueness results}
The  following proposition  is the  central technical  result  used to
establish   the    existence   and   uniqueness    of   solutions   of
Equation~\eqref{nlsde}, as well  as the mean-field convergence result.
For this last purpose, it  is convenient (but not essential) to obtain
a bound $C(t)$ which does not depend on $\Vert u'\Vert_t$.

\begin{proposition}\label{lp}
For $1\le k\le K$, 
let the functions $a_k$ be bounded and $a_k$, $b_k$, $\lambda_k$, and $\mu_k$ be Lipschitz, and the laws $\alpha_k$ have a first moment $m_k = \int w \,\alpha_k(dw)<\infty$, 
and let $\mathcal{A}_k$ be Poisson point processes on $\R_+^3$ with 
intensity measure $\alpha_k(dw)\,dz\,dt$ and 
$\mathcal{N}_k$ and $\mathcal{D}_k$ be Poisson point processes on $\R_+^2$ with intensity measure $dz\,dt$, all forming an independent family.
For any $u=(u(t), {t \ge0})$ in $C(\R_+,\R^J)$ and $(\R_+\cup\{-1\})^K$-valued random 
variable $X_0$ independent of the Poisson point processes, let
\[
\phi(X_0,u)=(X(t), {t \ge0})=((X_k(t), {t\geq 0}), {1\leq k\leq K})
\]
be the solution starting at $X(0)=X_0$
of the stochastic differential equation
\begin{equation} \label{sdeparam} \tag{$\mathcal{E}_u$}
\left\{
\begin{aligned}
& dX_k(t) =
\ind{X_k(t-) = -1}
\int (1+w) \ind{0 < z < \lambda_k\left(u(t)\right)}
\,\mathcal{A}_k(dw,dz,dt)
\\ 
& \hphantom{dX_k(t) =}
\;
+ \ind{X_k(t-)\ge 0} \biggl[ a_k\left(X_k(t-), u(t)\right)\,dt 
\\ 
& \hphantom{dX_k(t) =}
\qquad 
- (1-r_k) X_k(t-) 
\int \ind{0 < z < b_k\left(X_k(t-), u(t)\right)}
\,\mathcal{N}_k(dz,dt)
\\
& \hphantom{dX_k(t) =}
\qquad 
- (1 + X_k(t-))
\int \ind{0 < z < \mu_k\left(X_k(t-), u(t)\right)}
\,\mathcal{D}_k(dz,dt)
\biggr]\,,
\\
&1\le k \le K\,.
\end{aligned}
\right.
\end{equation}

Let $u=(u(t), {t \ge0})$ and $u'=(u'(t), {t \ge0})$ be in $C(\R_+,\R^J)$, and
initial values $X_0$ and $X'_0$ be independent of the Poisson point processes.
Then, for all $t\ge0$,
\begin{equation}\label{esteq}
\E\left[\|\phi(X_0,u)-\phi(X'_0,u')\|_t \,\Big|\, X_0, X'_0 \right]
\leq \left[\|X_0-X'_0\| +\int_0^tC(s)\|u-u'\|_s\,ds\right] e^{tC(t)}
\end{equation}
where $C(t) = A(1 + t + \Vert u\Vert_t + \Vert X_0\Vert+\Vert X'_0\Vert )$ for 
some constant $A<\infty$ depending only
on $a_k$, $b_k$, $\lambda_k$, $\mu_k$, and $m_k$, $1\le k \le K$.
\end{proposition}

\begin{proof}
Any solution $(Y(t), {t\ge0})=(Y_k(t), {1\le k \le K}, {t\ge0})$ of \eqref{sdeparam}, for any $u$ and initial condition, satisfies the \emph{a priori} affine growth bound  
\begin{equation}\label{upbd}
\E^0\bigl[ Y_k(t)^+ \bigr] \leq \max(Y_k(0)^+ , m_k) +\lVert a_k\rVert t\,,
\end{equation}
which notably allows to prove, classically under the present assumptions, existence and uniqueness 
of solutions for \eqref{sdeparam} recursively from jump instant to next jump instant, since it implies that these cannot accumulate. 
Let 
\[
X(t)=\phi(X_0,u)(t)\,,
\quad
X'(t)=\phi(X'_0,u')(t)\,, 
\qquad
t\ge0\,,
\] 
and $\E^0$ denote the conditional expectation given $X_0$ and $X'_0$.
Then
\begin{equation}
\label{corin}
\|X_k-X_k'\|_t \leq |X_k(0)-X_k'(0)| + I_1(t) + I_2(t) + I_3(t) + I_4(t)
\end{equation}
for
\begin{align*}
I_1(t) &= 
\int_0^t \int 
(1+w) \bigl|\ind{X_k(s-) = -1} \ind{0 < z < \lambda_k(u(s))}
\\
& \kern30mm
- \ind{X'_k(s-) = -1} \ind{0 < z < \lambda_k(u'(s))} 
\bigr|
\,\mathcal{A}_k(dw,dz,ds)\,,
\\
I_2(t) &= 
\int_0^t \bigl| 
\ind{X_k(s-)\ge 0} a_k(X_k(s), u(s))
- \ind{X'_k(s-)\ge 0} a_k(X'_k(s), u'(s))
\bigr|\,ds\,,
\\
I_3(t) &= 
\int_0^t \int \bigl|
X_k(s-)^+ \ind{0 < z < b_k(X_k(s-), u(s))}
\\
& \kern30mm
- X'_k(s-)^+ \ind{0 < z < b_k(X'_k(s-), u'(s))}
\bigr|\,\mathcal{N}_k(dz,ds)\,,
\\
I_4(t) &= 
\int_0^t \int \bigl| 
(1 + X_k(s-)) \ind{0 < z < \mu_k(X_k(s-), u(s))}
\\
& \kern30mm
- (1 + X'_k(s-)) \ind{0 < z < \mu_k(X'_k(s-), u'(s))}
\bigr|\,\mathcal{D}_k(dz,ds)\,,
\end{align*}
where the fact that $\ind{x\ge0}x = x^+$ and $\ind{x\ge0}(1+x) =1+x$
for $x$ in  $\R_+\cup\{-1\}$ is used for regularizing some integrands.

Compensating the Poisson point process $\mathcal{A}_k$ and the inequality 
\begin{equation}
\label{indlip}
\bigl|\ind{x = -1} - \ind{x' = -1}\bigr|
= \bigl|\ind{x \ge 0} - \ind{x' \ge 0}\bigr|
\le |x - x'|\,,
\;\;
x,x'\in \R_+\cup\{-1\}\,,
\end{equation}
(the functions $\mathbbm{1}_{\R_+}$ and $\ind{-1}$ are $1$-Lipschitz on $\R_+ \cup \{-1\}$)
yield 
\begin{align*}
\E^0[I_1(t)] 
&=
\E^0\biggl[
\int_0^t \int 
(1+w) \bigl|\ind{X_k(s) = -1} \ind{0 < z < \lambda_k(u(s))}
\\
& \kern20mm
- \ind{X'_k(s) = -1} \ind{0 < z < \lambda_k(u'(s))} 
\bigr|
\,\alpha_k(dw)\,dz\,ds
\biggr]
\\
&\le
(1+m_k)\E^0\biggl[
\int_0^t \int \Bigl[
\bigl|\ind{X_k(s) = -1} - \ind{X'_k(s) = -1}\bigr| \ind{0 < z < \lambda_k(u(s))}
\\
& \kern20mm
+ \ind{X'_k(s) = -1} \bigl|\ind{0 < z < \lambda_k(u(s))} - \ind{0 < z < \lambda_k(u'(s))} 
\bigr| \Bigr]\,dz\,ds
\biggr]
\\
&\le
(1+m_k)\E^0\biggl[
\int_0^t \int \Bigl[
\bigl|X_k(s) - X'_k(s)\bigr| \lambda_k(u(s))
+ \bigl|\lambda_k(u(s)) - \lambda_k(u'(s))\bigr| \Bigr]
\,ds
\biggr]\,.
\end{align*}

The singularities due to the cemetery state $-1$ now cause further difficulties. 
A \emph{key step} in the proof is that, since
the functions $a_k$, $b_k$, and $\mu_k$, defined on $\R_+\times \R_+^J$,
can be extended \emph{arbitrarily} to $(\R_+\cup\{-1\}) \times \R_+^J$ without altering 
$I_2$, $I_3$, and $I_4$, we may introduce adequate Lipschitz extensions for such functions $f$ by setting $f(-1,u) = f(0,u)$, \emph{i.e.},
\[
f(x,u) = f(x^+,u)\,,
\qquad
(x,u) \in (\R_+\cup\{-1\}) \times \R_+^J\,.
\]
Since $x \mapsto x^+$ is $1$-Lipschitz, the extensions are Lipschitz
on $(\R_+\cup\{-1\}) \times \R_+^J$ with same Lipschitz coefficients 
as the original function on $\R_+ \times \R_+^J$.
(Note that the extension vanishing on $\{-1\}\times \R^J$
may well not be Lipschitz on $(\R_+\cup\{-1\}) \times \R_+^J$, 
see for instance $(x,u)\in \R_+ \times \R^J \mapsto x+u$.)

In the sequel, we use these adequate Lipschitz extensions of $a_k$, $b_k$, and $\mu_k$, which play an important role in regularizing the integrands, and invisibly simplify greatly the computations.
Simple computations and \eqref{indlip} yield 
\begin{align*}
\E^0[I_2(t)] 
&=
\E^0\biggl[
\int_0^t 
\bigl|\ind{X_k(s) \ge 0} a_k(X_k(s), u(s)) - \ind{X'_k(s) \ge 0} a_k(X'_k(s), u'(s)) \bigr|
\,ds
\biggr]
\\
&\le
\E^0\biggl[
\int_0^t \int \Bigl[
\bigl|\ind{X_k(s) \ge 0} - \ind{X'_k(s) \ge 0}\bigr| a_k(X_k(s), u(s))
\\
& \kern20mm
+ \ind{X'_k(s) \ge 0} \bigl|a_k(X_k(s), u(s)) - a_k(X'_k(s), u'(s)) \bigr| 
\Bigr]\,dz\,ds
\biggr]
\\
&\le
\E^0\biggl[
\int_0^t \int \Bigl[
\bigl|X_k(s) - X'_k(s)\bigr| a_k(X_k(s), u(s))
\\
& \kern20mm
+ \bigl|a_k(X_k(s), u(s)) - a_k(X'_k(s), u'(s)) \bigr| 
\Bigr]\,dz\,ds
\biggr]\,,
\end{align*}
compensating the Poisson point process $\mathcal{N}_k$ yields
\begin{align*}
\E^0[I_3(t)] 
&=
\E^0\biggl[
\int_0^t \int 
 \bigl|X_k(s)^+ \ind{0 < z < b_k(X_k(s), u(s))}
- X'_k(s)^+ \ind{0 < z < b_k(X'_k(s), u'(s))} 
\bigr|\,dz\,ds
\biggr]
\\
&\le
\E^0\biggl[
\int_0^t \int \Bigl[
\bigl|X_k(s)^+ - X'_k(s)^+\bigr| \ind{0 < z < b_k(X_k(s), u(s))}
\\
& \kern20mm
+ X'_k(s)^+ \bigl|\ind{0 < z < b_k(X_k(s), u(s))} - \ind{0 < z < b_k(X'_k(s), u'(s))} \bigr| 
\Bigr]\,dz\,ds
\biggr]
\\
&\le
\E^0\biggl[
\int_0^t \Bigl[
\bigl|X_k(s) - X'_k(s)\bigr| b_k(X_k(s), u(s))
\\
& \kern20mm
+ X'_k(s)^+ \bigl|b_k(X_k(s), u(s)) - b_k(X'_k(s), u'(s)) \bigr| 
\Bigr]
\,ds
\biggr]
\end{align*}
and similarly, replacing $x^+$ by $1+x$ in the computations,
\begin{align*}
\E^0[I_4(t)] 
&\le
\E^0\biggl[
\int_0^t \Bigl[
\bigl|X_k(s) - X'_k(s)\bigr| \mu_k(X_k(s), u(s))
\\
& \kern20mm
+ (1+X'_k(s))\bigl|\mu_k(X_k(s), u(s)) - \mu_k(X'_k(s), u'(s)) \bigr| 
\Bigr]
\,ds
\biggr]\,.
\end{align*}

Going back to \eqref{corin}, the bounds on $\E^0[I_1(t)]$, $\E^0[I_2(t)]$, $\E^0[I_3(t)]$, and $\E^0[I_4(t)]$, the affine growth bound \eqref{upbd}, and the Lipschitz bounds which
remain valid for the \emph{extensions} of $a_k$, $b_k$ and $\mu_k$ on $(\R_+\cup\{-1\})\times\R_+^J$), yield
\begin{equation*}
\E^0\bigl[\|X_k-X_k'\|_t \bigr] 
\leq 
|X_k(0)-X_k'(0)|
+\int_0^t C_k(s) 
\left(\|u(s)-u'(s)\|+\E^0\bigl[\|X_k-X_k'\|_s\bigr]\right)ds
\end{equation*}
where $C_k(s) = A_k(1 + s + \Vert u\Vert_s + \Vert X_0\Vert+\Vert X'_0\Vert)$ for 
some constant $A_k<\infty$ depending only
on $a_k$, $b_k$, $\lambda_k$, $\mu_k$, and $m_k$, $1\le k \le K$.
The Gronwall Lemma then yields
\begin{align*}
\E^0\bigl[\|X_k-X_k'\|_t \bigr] 
&\leq 
\left[\left|X_k(0){-}X_k'(0)\right| +\int_0^t C_k(s) \|u(s){-}u'(s)\|\,ds\right]e^{\int_0^tC_k(s)\,ds}
\\
&\leq \left[\rule{0mm}{4mm}|X_k(0)-X_k'(0)| +\int_0^t C_k(s)\|u-u'\|_s\,ds\right]e^{tC_k(t)}
\end{align*}
and the proof of the proposition follows.
\end{proof}

One  of the  main motivations  of  this study  is to  obtain  results which  are valid  for
functions $\mu_k$ and $b_k$ of the form given 
in~\eqref{outputprop} and \eqref{bafopa}. Such functions may have quadratic behavior, and Proposition~\ref{lp} must be adapted to this case,
which replaces $\Vert X_0\Vert+\Vert X'_0\Vert$ by $\Vert X_0\Vert^2+\Vert X'_0\Vert^2$ in the bound for $C(t)$, but for simplicity this will be stated only when needed.

In order to control the evolution of $(W(t), {t\ge0})$ and
describe its  stationary behavior, the simple assumption that initial conditions be uniformly
bounded is not satisfactory.  
For these purposes, exponential  and  Gaussian  moment  assumptions  are introduced.

\bigskip\noindent
\textbf{Condition (C)}  
{\it
It  is  said  to  hold for a  family 
$\{X_0^\theta, {\theta\in  \Theta}\}$
of $(\R_+\cup\{-1\})^K$-valued r.v., 
for the functions $b_k : \R_+ \times \R_+^J \to \R_+$ and $\mu_k : \R_+ \times \R_+^J \to \R_+$
for $1\le k\le K$, and  for $\eps>0$, 
when at least one of the two following conditions is satisfied:
\begin{enumerate}
\item 
The functions $b_k$ and $\mu_k$ are Lipschitz for 
all $1\leq k\leq K$, 
and the $X_0^\theta$ for $\theta\in  \Theta$   
have a uniform exponential moment of order $\eps$:
\[
\sup_{\theta\in\Theta} \E\bigl[\exp(\eps \| X_0^\theta\|) \bigr] < \infty\,.
\]
\item 
There are Lipschitz functions $\beta_k: \R_+^J \to \R_+$ and $\nu_k: \R_+^J \to \R_+$ such that
\[
b_k (w, u) = w \beta_k (u)\,,
\quad
\mu_k (w, u) = w \nu_k (u)\,,
\qquad
1\leq k\leq K\,,
\]
and the $X_0^\theta$,  $\theta\in  \Theta$,   have a uniform Gaussian moment of order $\eps$:
\[
\sup_{\theta\in\Theta} \E\bigl[\exp\bigl(\eps \| X_0^\theta \|^2\bigr) \bigr]<\infty\,.
\]
\end{enumerate}
}
\noindent
The following theorem establishes the existence and uniqueness result for Equation~\eqref{nlsde}. 

\begin{theorem}\label{eunlsde}
If the functions $a_k : \R_+ \times  \R_+^J \rightarrow \R_+$ are bounded
and Lipschitz and $\lambda_k : \R_+^J \rightarrow \R_+$ are Lipschitz, $1\leq k\leq K$,
and if Condition~(C) holds for the $(\R_+\cup\{-1\})^K$-valued r.v.\ $W_0$, 
the functions $b_k$ and $\mu_k$ for $1\leq k\leq K$, and $\eps>0$, then there is
pathwise existence  and uniqueness  of the solution $(W(t), {t\ge0})$ of the nonlinear
stochastic differential equation~\eqref{nlsde} starting at $W_0$. 

In this  case, the solution depends  continuously on the initial  condition in the
following way: if $(W(t), {t\ge0})$ and $(W'(t), {t\ge0})$ are solutions of\eqref{nlsde} with
respective initial conditions $W_0$ and $W_0'$ having the same moment in Condition~(C),
then for $T\geq 0$ there exists a constant $A_T$ such that  
\begin{multline}\label{condinit}
\E\left[\|W-W'\|_T\right] := \E\biggl[\sup_{s\le T}\|W(s)-W'(s)\|\biggr]
\\ 
\leq A_T
\E\left[\|W_0-W_0'\|e^{\eps(\|W_0\|^\ell+\|W_0'\|^\ell)/2)}\right]
e^{\E[\exp(\eps(\|W_0\|^\ell+\|W_0'\|^\ell))]}
\end{multline}
where $\ell=1$ or $2$ depending on whether the moment in Condition~(C) is exponential
or Gaussian.
\end{theorem}

\begin{proof}
The proof follows \emph{mutatis mutandis} the proof of Graham-Robert~\cite[Theorem 4.2]{Graham:09},
substituting Proposition~\ref{lp} to \cite[Proposition 4.1]{Graham:09}, and we give some details
about this.
If (1) of Condition~(C) holds, then \eqref{esteq} (given in
Proposition~\ref{lp}) is substituted to \cite[(4.3)]{Graham:09}
(given in \cite[Proposition 4.1]{Graham:09}). If~(2) of Condition~(C) holds, then Proposition~\eqref{lp} and its proof must be adapted
as in the end of the proof of \cite[Theorem 4.2]{Graham:09}: the term $\E^0[I_3(t)]$ must be bounded using
\begin{multline*}
\bigl|b_k(X_k(s),u(s))-b_k(X_k'(s),u'(s))\bigr|
=
\bigl| X_k(s) \beta_k({u}(s)) - X'_k(s) \beta_k({u}'(s)) \bigr|
\\\leq
X_k(s) | \beta_k({u}(s)) - \beta_k({u}'(s)) |
+
\beta_k({u}'(s)) | X_k(s) - X'_k(s) |
\end{multline*}
and similarly for $\E^0[I_4(t)]$, and the rest of the proof is analogous, 
the supplementary multiplications by $X_k(s)$ and $X'_k(s)$ being handled by using the Gaussian moment assumption. 
\end{proof}


\section{Mean-field limit theorem for converging initial data}\label{secMF}\setcounter{equation}{0}

It is said that multi-indices ${N}=(N_k, {1\le k \le K})$ go to infinity, denoted by
${N} \to \infty$,
when 
\[
\min_{1\le k \le K}N_k \to \infty\,.
\] 
This is the case in the mean-field scaling, where \eqref{e:mfar} gives the relative growth rate of the numbers of class $k$ users. In this context, the
notions of exchangeability and chaoticity play a fundamental role. These properties are classical
for single-class systems, see Aldous~\cite{Aldous} and Sznitman~\cite{Sznitman} for example, 
but need to be extended to the present multi-class network. 

\begin{definition}
The family of r.v. 
$(X_{n,k}, {1\le n \le N_k}, {1\le k \le K})$
is \emph{multi-exchangeable} if its law is invariant under
permutation of the indexes \emph{within} the classes:
for any permutations $\sigma_k$ of $\{1,\ldots,N_k\}$ for $1\le  k\le K$, 
there holds the equality of laws 
\[
\mathcal{L}(X_{\sigma_k(n),k}, {1\le n \le N_k}, {1\le k \le K})
=
\mathcal{L}(X_{n,k}, {1\le n \le N_k}, {1\le k \le K})\,.
\]
A family $\bigl(X^{N}_{n,k}, {1\le n \le N_k}, {1\le k \le K}\bigr)$ 
of multi-class random variables
indexed by ${N}=(N_k, {1\le k \le K}) \in\N^K$ is  $P_1\otimes\cdots\otimes P_K$-\emph{multi-chaotic}, where each $P_k$ is a probability measure, if
for any $m\geq 1$ there holds the weak convergence of laws
\[
\lim_{{N} \to \infty} 
\mathcal{L}\bigl(X^{N}_{n,k}, {1\le n \le m}, {1\le k \le K}\bigr)
= P_1^{\otimes m} \otimes \cdots \otimes P_K^{\otimes m}\,.
\] 
\end{definition}

Hence, such a family of systems is multi-chaotic if and only if it becomes asymptotically
independent with particles of class~$k$ having law $P_k$.
A surprising result, not used in this  paper, is that a  family of multi-ex\-chan\-geable
multi-class systems is  multi-chaotic \emph{if and only if} the \emph{restriction} to each class
is chaotic, see Graham~\cite[Theorem~3]{Graham:08}.

The following theorem is the main mean-field convergence result.
The underlying topology on the corresponding functional space is uniform convergence
on compact sets.  

\begin{theorem} \label{mulex}
It is assumed that:
\begin{enumerate}
\item 
the mean-field regime \eqref{e:mfar} holds:
the multi-indices $N=(N_1, \dots, N_K)$ go to infinity so that 
\[
\lim_{N\to +\infty} \frac{N_k}{N_1+\cdots+N_K}=p_k\,,
\qquad
1\leq k\leq K\,,
\]
\item 
the $\R_+ \cup \{-1\}$-valued random variables 
\[
\bigl(W_{n,k}^{N}(0), {1\le n \le N_k}, {1\le k \le K}\bigr)
\]
are multi-exchangeable and $P_{1,0}\otimes\cdots\otimes P_{K,0}$-multi-chaotic, 
where $P_{k,0}$ is a probability distribution on
$\R_+ \cup \{-1\}$ for $1\leq k\leq K$, 
\item 
the functions $a_k : \R_+ \times \R_+^J \rightarrow \R_+$ are bounded and Lipschitz, 
the functions $\lambda_k : \R_+^J \rightarrow \R_+$ are Lipschitz, and
Condition~(C) holds for the random variables
\[
\{W_1^{N}(0), {N\in\N^K}\} = \{(W_{1,k}^{N}(0), {1\le k \le K}), {N\in\N^K}\}\,,
\]
the functions $b_k$ and $\mu_k$ for $1\le k \le K$, and $\eps>0$.
\end{enumerate}
Then, for $1\leq k\leq K$, in the sense of processes
\[
\lim_{N\to+\infty} \E\biggl|\frac{1}{N_k}\sum_{n=1}^{N_k} W_{n,k}^N(t) -\E(W_k(t))\biggr|=0
\]
and the family  of processes
\[
\bigl((W_{n,k}^{N}(t), {t\ge 0} ), {1\le n \le N_k}, {1\le k \le K}\bigr)
\]
given by the solutions of the SDE~\eqref{sde} starting at  the initial  conditions 
$\bigl(W_{n,k}^{N}(0)\bigr)$  is multi-exchangeable
and  $P_W$-multi-chaotic,  where  $P_W = P_{W_1}\otimes\cdots\otimes P_{W_K}$ is  the 
law of the solution 
\[
(W(t), {t\geq 0})=((W_k(t), {t\geq 0}), {1\leq k\leq K})
\]
of the nonlinear SDE  \eqref{nlsde}  with initial distribution $P_{1,0}\otimes
\cdots \otimes P_{K,0}$.  

In  particular,
$(W_{n,\cdot}^{N}(t), t\geq 0) := ((W_{n,k}^{N}(t), {1\le k  \le K}), {t\geq 0})$ 
converges in distribution to $(W(t), t\geq 0) = ((W_k(t), {1\le k  \le K}), {t\geq 0})$
when $N\to\infty$, all $n\geq  1$.

\end{theorem}
\begin{proof}
As in the proof for Theorem~\ref{eunlsde}, the proof follows the proof of
Graham-Robert~\cite[Theorem 5.1]{Graham:09}, substituting
Proposition~\ref{lp}  to \cite[Proposition 4.1]{Graham:09}
if (1) of Condition~(C) holds, and adapting 
appropriately Proposition~\ref{lp} if (2) of Condition~(C) holds.
\end{proof} 

The following is a simple consequence of the above result and of exchangeability.
 
\begin{corollary}
Under the assumptions and notations of Theorem~\ref{mulex}, 
the convergence in law of the empirical distributions
\[
\lim_{{N}\to\infty} \Lambda_k^{N} = P_{W_k}\,,
\qquad
1\le k \le K\,,
\]
holds for the weak topology on $\mathcal{P}(\DD(\R_+, \R_+\cup\{-1\}))$ with  $\DD(\R_+, \R_+\cup\{-1\})$
endowed with the Skorohod topology.

\end{corollary}

\section{Equilibrium behavior}
\label{secEqui}

\subsection{Stationary distributions for the nonlinear SDE}
The basic equilibrium properties of $(W(t), {t\geq 0})=((W_k(t), {t\geq 0}),
{1\leq  k\leq  K})$, solution  of  the  nonlinear  SDE~\eqref{nlsde}
describing the limit user evolution, are now succinctly studied,
in  particular its stationary distributions. 

Note  that, due  to the non-linearity, $(W(t),  {t\geq 0})$ is \emph{not} a homogeneous Markov
process, and therefore the classical theory concerning the convergence
toward equilibrium of Markov processes does not apply.

The coordinates of the process $(W(t), {t\geq 0})$ evolve independently,
and are coupled through the presence in the coefficients
of Equation~\eqref{nlsde} of the load vector $(u_W(t), {t\geq 0})$ given by
\begin{equation}\label{fpe}
u_W(t)=(u_{W,j}(t), 1\leq j\leq J)\,,
\qquad
u_{W,j}(t) = \sum_{k=1}^K A_{jk} p_k \E(W_k(t)^+)\,.
\end{equation}
If a stationary distribution exists, and is taken as the initial distribution, then
$u_W(t)$ is constant in $t$, and hence
$(W(t), {t\geq 0})$ is a stationary process with independent coordinates.

One begins with the analysis of the equilibrium of a single process $(W_k(t), {t\geq 0}))$ when the load vector
$u_W(t)$ is replaced by a constant $u$, then examines the fixed-point equation
resulting from the coupling of the coordinates through Relation~\eqref{fpe}, which in equilibrium does not depend on $t$.

\subsection{Invariant measures for a generic process}

Throughout this section a load vector $u = (u_j, {1\le j \le J})\in\R_+^J$ is fixed
without further mention, and simplified notations
are used, dropping the index $k$ from notations. Normal notation will be resumed in the following section.

One studies the invariant measures for the classic Itô-Skorohod SDE
on $\R_+\cup\{-1\}$
\begin{align}
dW(t) & =
\ind{W(t-) = -1}
\int (1+w) \ind{0 < z < \lambda\left(u\right)}
\,\mathcal{A}(dw, dz, dt)\notag 
\\ \qquad 
& \quad
+ \ind{W(t-)\ge 0} \Bigl[ a\left(W(t-), u\right)\, dt 
- (1-r) W(t-) 
\int \ind{0 < z < b\left(W(t-), u\right)}
\,\mathcal{N}(dz, dt)\notag 
\\
& \qquad 
- (1 + W(t-))
\int \ind{0 < z < \mu\left(W(t-), u\right)}
\,\mathcal{D}(dz, dt)
\Bigr]
\label{sdeg}
\end{align}
driven by Poisson point processes ${\cal A}$ on $\R_+^3$  with intensity measure
$\alpha(dw)dzdt$, and ${\cal D}$ and ${\cal N}$ on $\R_+^2$ with intensity measure $dzdt$.
An important element of the study will be the  SDE on $\R_+$, corresponding to the evolution 
of a permanent connection,
\begin{equation}\label{perm}
dV(t) =
 a\left(V(t-), u\right)\, dt 
- (1-r) V(t-) 
\int \ind{0 < z < b\left(V(t-), u\right)}
\,\mathcal{N}(dz, dt)\,. 
\end{equation}
For definitions and results on Harris ergodicity, see
Nummelin~\cite{Nummelin:02}, Asmussen~\cite{Asmussen}. 
Actually, classical positive recurrent state techniques are used.
\begin{theorem} \label{th-ergoV}
Assume that 
\begin{enumerate}
\item the function $x\mapsto a(x,u)$ is Lipschitz bounded, 
and $\inf\{a(x,u):x\ge0\}>0$,
\item the function $x\mapsto b(x,u)$ is locally bounded, and 
$\inf\{b(x,u):x\ge x_0\} >0$ for some $x_0\ge0$.
\end{enumerate}
Then, for given initial conditions, there exists a unique solution $(W(t), {t\ge0})$ 
of \eqref{sdeg}, and a unique solution $(V(t), {t\ge0})$ 
of \eqref{perm}, and this defines Markov processes.  
Moreover, $(V(t), {t\ge0})$ has a positive recurrent state and 
hence a stationary distribution $\pi^V$, 
and is in particular Harris ergodic.
\end{theorem}

\begin{proof}
The existence and uniqueness result of solution and the Markov property are classical.
The hitting time of $x_0$ by $(V(t), {t\ge0})$ is defined by
\[
T_{x_0}=\inf\{t>0: V(t)=x_0, \exists s<t, V(s) \neq x_0\}\,.
\]
Let $\eta:=\inf\{b(x,u):x\ge x_0\}$, and
$(\widetilde{V}(t), {t\ge0})$ be the solution of the SDE\begin{equation}\label{aimdd}
d\widetilde{V}(t) =
 a(\widetilde{V}(t-), u)\, dt 
- (1-r) \widetilde{V}(t-) 
\int \ind{0 < z < \eta}
\,\mathcal{N}(dz,dt)
\end{equation}
(existence and uniqueness are classical).
If $\widetilde{V}(0) \geq V(0) \geq x_0$ then it is a simple matter to construct a coupling 
of $(\widetilde{V}(t), {t\ge0})$ and $(V(t), {t\ge0})$
such that the relation  $\widetilde{V}(t) \geq V(t)$ holds as long as $V(t)\geq x_0$.  
The Markov  process $(\widetilde{V}(t))$ is Harris ergodic; see Dumas \emph{et al.}~\cite{Dumas:06} for the existence
of a stationary distribution satisfying  the coupling  property, and
Chafa\"i \emph{et al.}~\cite{Chafai} for more on long-time convergence.  
The  hitting  time  below  $x_0$  is  therefore  integrable  for
$(\widetilde{V}(t))$, and hence for $(V(t))$.

Now if $V(0) < x_0$ then, since there are no upward jumps and
$x\mapsto a(x,u)$ is lower bounded away from $0$ and $x\mapsto b(x,u)$ is
bounded on the compact set $[0,x_0]$, it is a simple matter to show that 
the hitting time of $x_0$ by $(V(t))$ is integrable. One therefore
concludes that $\E_{x_0}(T_{x_0})<\infty$, and $x_0$ is positive recurrent,
and classical results conclude the proof.
\end{proof}

Note that $(W(t), {t\ge0}))$ (as defined in this section) is an homogeneous Markov  process, and 
can be seen as a solution  of the SDE~\eqref{perm} killed at some random instants
and  regenerated  after an  exponentially  distributed  duration  of time  with  parameter
$\lambda(u)$.
Because of this regenerative structure, the stationary distribution of the SDE~\eqref{perm} will not come into play as it did in Graham and 
Robert~\cite{Graham:09}.

\begin{theorem}\label{th-invmeaW}
Let the functions $a$ and $b$ satisfy the assumptions of
Theorem~\ref{th-ergoV}, and moreover $\lambda(u)>0$ and $\int \mu(x,u)\pi^V(dx) >0$.
Let  $(V(t), {t\ge0}))$ be the solution of~\eqref{perm} starting at $V(0)$ 
with law $\alpha(dw)$, and $\gamma$ be the positive measure on $\R_+\cup\{-1\}$ defined by 
\begin{equation}
\label{invmeas}
\int f(x)\,\gamma(dx)= f(-1) + \lambda(u) \int_0^{+\infty} 
\E\left(f(V(t))\exp\left(-\int_0^{t} \mu(V(s),u)\,ds\right)\right)\,dt
\end{equation}
for all non-negative Borel functions $f$ on $\R_+\cup\{-1\}$. 
Then $\gamma$ is the unique invariant measure of the Markov process $(W(t), {t\ge0}))$
defined by~\eqref{sdeg}: for all $t\ge0$,
\[
\int \E_x(f(W(t)))\gamma(dx)=\int f(x)\gamma(dx)
\]
for all non-negative Borel functions $f$ on $\R_+\cup\{-1\}$. 
\end{theorem}
\begin{proof}
Let the stopping time
\[
\tau_{-1}=\inf\{t>0: W(t)=-1, \exists s<t, W(s) \neq -1\}
\]
be the cycle time of  $(W(t), {t\ge0})$ between
two visits to $-1$.
Classically, if $\P(\tau_{-1}=\infty)=0$ then 
it is well known that 
the measure $\hat{\gamma}$ defined by 
\begin{equation}\label{aux}
\int f(x)\,\hat{\gamma}(dx) = \E_{-1}\left(\int_0^{\tau_{-1}} f(W(s))\,ds\right)
\end{equation}
for all non-negative Borel functions $f$ on $\R_+\cup\{-1\}$ 
is the unique invariant measure for $(W(t), {t\ge0})$,
see \emph{e.g.} Asmussen~\cite[Proposition~3.2]{Asmussen}, 
Robert~\cite[Proposition~8.12]{Robert:08}. 

In the following, $W(0)=-1$.
Then $(W(t), {t\ge0})$ remains in
$-1$ for an exponentially distributed duration $E_{\lambda(u)}$ with parameter $\lambda(u)$,
after which it jumps to an $\alpha$-distributed state in $\R_+$. Then, its excursion in $\R_+$
has same distribution as $(V(t), {t\ge0})$ until it returns to $-1$.  
Thus, $\hat{\gamma}(\{-1\})=1/\lambda(u)$, and $\tau_{-1} = E_{\lambda(u)} +\tau$ 
where $\tau$ is the first instant of an inhomogeneous Poisson process  on $\R_+$
with rate function $(\mu(V(t),u), {t\ge0})$, and 
\[
\P(\tau >t\mid V(s), 0\leq s\leq t)=\exp\left(-\int_0^t \mu(V(s),u)\,ds\right).
\] 
Since $(V(t), {t\ge0})$ is Harris ergodic and $\int \mu(x,u)\pi^V(dx) >0$, 
the corresponding ergodic theorem yields
\[
\lim_{t\to +\infty} \int_0^t \mu(V(s),u)\,ds = +\infty\,,
\;\;\text{a.s.},
\]
so that $\P(\tau = \infty)=0$, and thus $\P(\tau_{-1}=\infty)=0$
implying that $\hat{\gamma}$ is indeed the unique invariant measure.  
From the above,
\begin{align*}
\E\left(\int_0^{\tau} f(V(t))\,dt\right)
&=\E\left(\int_0^{+\infty}f(V(t))\ind{\tau>t}\,dt\right)\\
&=\E\left(\int_0^{+\infty} f(V(t))\P(\tau>t \,|\, V(s), 0\leq s\leq t)\,dt\right)\\
&= \int_0^{+\infty} \E\left(f(V(t))\exp\left(-\int_0^{t} \mu(V(s),u)\,ds\right)\right)\,dt
\end{align*}
so that $\gamma$ defined by \eqref{invmeas} is equal to $\lambda(u)\hat{\gamma}$. 
\end{proof}

If the measure $\gamma$ has a finite mass then there is a unique stationary
distribution given by $\pi:=\gamma/(1+\gamma(\R_+))$, else there is none.

\begin{remark}
The measure $\gamma$ is finite on all Borel functions $f$ such that, for some $C>0$, $|f(x)|\leq
C\mu(x,u)$ for any $x\geq 0$. Indeed, since $t\to V(t)$  is continuous almost everywhere
for the Lebesgue measure on $\R_+$,  then
\begin{multline*}
\int_0^{+\infty}|f(V(t))|\exp\left(-\int_0^{t} \mu(V(s),u)\,ds\right)\,dt \\ \leq 
C\int_0^{+\infty}\mu(V(t),u)\exp\left(-\int_0^{t} \mu(V(s),u)\,ds\right)\,dt \leq C.
\end{multline*}
In particular, if $\inf_{x\ge0}\mu(x,u)>0$ then $\gamma$ has a finite mass. 
\end{remark}

\subsection{Fixed-point equations}

In this section, the study of the case of $K$ classes of users is resumed, and
regular notation is again used.
It is assumed
that the functions $a_k(\cdot, u)$ and $b_k(\cdot,u)$ satisfy the assumptions of
Theorem~\ref{th-ergoV}, and
that the process $(W_k(t), {1\le k \le K})$ has a stationary distribution $\pi$ on $\R^K$, 
and the corresponding load vector is denoted by $u=(u_j, {1\le j \le J})$. 

Because of the independence of the coordinates,
$\pi$ can be written as $\pi=\bigotimes_{k=1}^K\pi_{k,u}$ where $\pi_{k,u}$ is the stationary
distribution of the solution $(W_k(t), {t\ge0})$  of SDE~\eqref{sdeg} with coefficients 
$\lambda_k(u)$ and $r_k$ and functions $a_k(\cdot,u)$,
$b_k(\cdot,u)$ and $\mu_k(\cdot,u)$. By Theorem~\ref{th-invmeaW}, 
\begin{equation}\label{cont}
\left\{
\begin{aligned}
\int &f(x)\,\pi_{k,u}(dx) 
= \frac{1}{1+\lambda_k(u)Z_k(u)}f(-1)\\
&+ \frac{\lambda_k(u)}{1+\lambda_k(u)Z_k(u)}\int_0^{+\infty} \E\left(f(V_{k,u}(t))\exp\left(-\int_0^{t} \mu_k(V_{k,u}(s),u)\,ds\right)\right)\,dt\,,
\\
&\text{ with }Z_k(u) 
=\int_0^{+\infty} \E\left(\exp\left(-\int_0^{t} \mu_k(V_{k,u}(s),u)\,ds\right)\right)\,dt\,,
\end{aligned}
\right.
\end{equation}
for  all non-negative Borel functions $f$ on $\R\cup\{-1\}$, where
$V_{k,u}$  is the unique  solution
of the SDE~\eqref{perm} (with index $k$ added to the functions) with
initial distribution $\alpha_k(dw)$, and $Z_k$
is the appropriate normalization constant. 
Hence, by rewriting \eqref{fpe} one gets the following theorem.
 
\begin{theorem} 
Let the functions $a_k$ and $b_k$ and $\lambda_k$ and $\mu_k$ satisfy the assumptions of
Theorem~\ref{th-invmeaW} for any $1\leq k\leq K$ and $u\in\R_+^J$.
Then, any stationary distribution of the nonlinear
SDE~\eqref{nlsde} can be written as $\pi = \bigotimes_{k=1}^K\pi_{k,u}$, 
where $\pi_{k,u}$ is defined by Equation~\eqref{cont} and  $u=(u_j, {1\le j \le J})\in\R_+^J$ is a solution of the fixed point equation
\begin{equation}\label{fpem}
u_{j}= \sum_{k=1}^K A_{jk} p_k \frac{\lambda_k(u)}{1+\lambda_k(u)Z_k(u)}
\int_0^{+\infty}
\E\left(V_{k,u}(t)\exp\left(-\int_0^{t} \mu_k(V_{k,u}(s),u)\,ds\right)\right)\,dt
\end{equation}
where $V_{k,u}$  is the  solution
of the SDE~\eqref{perm} associated to $r_k$ and the functions $a_k(\cdot, u)$, $b_k(\cdot,u)$, $\mu_k(\cdot,u)$ with
initial distribution $\alpha_k(dw)$ and $Z_k(u)$ is the corresponding
normalizing constant given in Equation~\eqref{cont}.
\end{theorem}

The characterization of the invariant distributions involves the transitory behavior of the solution
$(V_{k,u}(t), {t\ge0})$, which leads to a more complex and less explicit expression
for the fixed-point equation in comparison to the case of permanent
connections investigated in Graham and  Robert~\cite{Graham:09}.

One concludes with two examples for the service rates $(\mu_k)$ under the assumptions of the
above theorem.
\subsubsection*{Constant service rate}
It is assumed that $\mu_k(x,u)=\mu_k(u)$ for $x\geq 0$ and
$u\in\R_+^J$. Equation~\eqref{fpem} is in this case
\[
u_{j}= \sum_{k=1}^K A_{jk} p_k \frac{\lambda_k(u)}{1+\lambda_k(u)/\mu_k(u)}\int_{0}^{+\infty}
\E\left[V_{k,u}(t)\right]e^{-\mu_k(u)t}\,dt, \quad 1\leq j\leq J.
\]
The integral in the right hand side of the above equation is related to the resolvent of
the Markov process $(V_{k,u}(t))$.
\subsubsection*{Linear service rate}
It is assumed that $\mu_k(x,u)=x\mu_k(u)$ for $x\geq 0$ and
$u\in\R_+^J$. Equation~\eqref{fpem} becomes
\[
u_{j}= \sum_{k=1}^K A_{jk} p_k \frac{\lambda_k(u)/\mu_k(u)}{1+\lambda_k(u)Z_k(u)},
\]
with 
\[
Z_k(u)=\int_{0}^{+\infty} \E\left(\exp\left(-\mu_k(u)  \int_0^{t}V_{k,u}(s)\,ds\right)\right)\,dt.
\]
\subsection*{Future Work}
As it may be seen from Relation~\eqref{fpem}, the explicit representation of the invariant
distribution involves the distribution of  the solution of the SDE~\eqref{perm} associated
to  a  permanent  connection.  If  the  equilibrium behavior  of  this  process  is  fully
understood,  its transient  characteristics are  not  well known.   See Chafa\"i  \emph{et
  al.}~\cite{Chafai} for the rate of convergence to equilibrium. The case of constant rate
shows that one should have an expression of the resolvent of this process.  This
is, in our view, an interesting challenging problem of this domain.

\providecommand{\bysame}{\leavevmode\hbox to3em{\hrulefill}\thinspace}
\providecommand{\MR}{\relax\ifhmode\unskip\space\fi MR }
\providecommand{\MRhref}[2]{%
  \href{http://www.ams.org/mathscinet-getitem?mr=#1}{#2}
}
\providecommand{\href}[2]{#2}

\end{document}